\documentclass[11pt]{article}

\usepackage[margin=1in]{geometry}
\usepackage{amsmath,amssymb,amsthm}
\usepackage{mathtools}
\usepackage{parskip}
\IfFileExists{newpxtext.sty}{\usepackage{newpxtext,newpxmath}}{\usepackage{mathpazo}}
\usepackage[hidelinks]{hyperref}
\usepackage{cleveref}
\numberwithin{equation}{section}

\theoremstyle{plain}
\newtheorem{theorem}{Theorem}[section]
\newtheorem{lemma}[theorem]{Lemma}
\newtheorem{claim}[theorem]{Claim}
\theoremstyle{definition}
\newtheorem{definition}[theorem]{Definition}
\theoremstyle{remark}
\newtheorem{remark}[theorem]{Remark}
\newtheorem{corollary}[theorem]{Corollary}

\DeclareMathOperator{\ch}{ch}
\DeclareMathOperator{\spa}{span}
\newcommand{\F}{\mathbb{F}}
\newcommand{\Pal}{\mathrm{Pal}}
\newcommand{\LID}{\mathrm{LID}}
\newcommand{\CoeffX}{\mathrm{Coeff}_X}
\newcommand{\rank}{\mathrm{rank}}
\newcommand{\coeff}{\mathrm{coeff}}
\newcommand{\coef}{\coeff}
\newcommand{\spn}{\mathrm{span}}
\newcommand{\mindeg}{\mathrm{mindeg}}
\newcommand{\B}{\mathcal{B}}

\newcommand{\h}[1]{\widehat{#1}}

\title{A Quadratic Lower Bound for Noncommutative Circuits}
\author{Pratik Shastri\thanks{IT University of Copenhagen, Copenhagen, Denmark. Supported by the Carlsberg Foundation grant CF25-1645. Email: \texttt{prasha@itu.dk}} \and Sophus Valentin Willumsgaard\thanks{Department of Computer Science, University of Copenhagen, Denmark. Supported by the European Research Council (ERC) under grant agreement no.~101125652 (ALBA). Email: \texttt{sophus.willumsgaard@di.ku.dk}}}
\date{November 2, 2026}

\begin{document}

\maketitle

\begin{abstract}
	We prove that every fan-in 2 noncommutative arithmetic circuit computing the palindrome polynomial $\Pal_{n,d}(X,Y)$ of degree $2d$ in $2n$ variables has size $\Omega(nd)$, which is tight up to a constant. In particular, when $d = n$ we obtain an $\Omega(n^2)$ lower bound, improving the previous lower bound of $\Omega(n \log n)$ by \cite{baur1983,strassen1973,nisan1991} and concurrent work of \(\Omega(n^{1.5})\) by \cite{raz2026}.

	The proof builds on and refines a previous work of the first author,
	by constructing for every gate a small set that spans a subset of the coefficients of the polynomial computed by that gate. We show that the structure of the set depends only on the subsets of multiplication gates along paths going from that gate to a leaf. From this description, we are then able to get an upper bound on this set based on the amount of product gates in the circuit.

	The proof is surprisingly short and simple to understand,
	considering that improving Nisan's \(\Omega(n \log n)\) lower bound for noncommutative circuits has been an open problem for
	decades.
\end{abstract}

\newpage

\section{Introduction}
An arithmetic circuit for a noncommutative polynomial $f$ is a directed acyclic graph (DAG) which represents a way to compute $f$, starting from variables and field constants (these label the leaves) and using addition and (noncommutative) multiplication as building blocks (these label the internal vertices). Its size is measured in terms of the number of vertices in the underlying DAG. The fan-in of a circuit is defined as the maximum in-degree of any of its vertices.

Proving lower bounds on the size of noncommutative arithmetic circuits computing explicit\footnote{We say a polynomial is explicit if there exists a polynomial time algorithm which can compute the coefficient of a monomial provided as an input. Since we will only deal with polynomials that have coefficients in $\{0,1\}$, this definition suffices for us.} polynomials is an important problem in arithmetic circuit complexity. For explicit polynomials on $n$ variables with degree $\Theta(n)$, the strongest known lower bounds for general circuits are of the form $\Omega(n\log n)$ \cite{strassen1973,baur1983}. While the eventual goal is to be able to prove superpolynomial or even exponential (in $n$) lower bounds, this seems out of reach at the present moment and the best known lower bound has remained stagnant at $\Omega(n\log n)$ for decades. The starting point for much of the area is the work of Nisan \cite{nisan1991}, who proved exponential lower bounds on the size of noncommutative formulas computing explicit polynomials. Formulas are circuits whose graph is a tree. His result also established an exponential separation between the power of formulas and circuits, but it did not yield similarly strong bounds for general circuits \cite{nisan1991}.

Subsequent work has explored noncommutative circuits from various angles. For instance, Hrube\v{s}, Wigderson and Yehudayoff \cite{hrubes2011} connected noncommutative circuit lower bounds to the classical sum-of-squares problem, while Carmosino, Impagliazzo, Lovett and Mihajlin \cite{carmosino2018} showed a hardness-amplification phenomenon. They show that mildly superlinear lower bounds for noncommutative circuits computing constant degree polynomials would already imply exponential lower bounds.

Restricted variants of noncommutative circuits have also been explored in the literature. A circuit is said to be homogeneous if it only ever computes homogeneous polynomials. Recently, Chatterjee and Hrube\v{s} \cite{chatterjee2023} showed that lower bounds better than $\Omega(n\log n)$ can be proved under the homogeneity restriction. In particular, they constructed an explicit polynomial $f$ on $n$ variables of degree $\Theta(n)$ such that any homogeneous noncommutative circuit with fan-in 2 computing $f$ requires size $\Omega(n^2)$. This gives a substantial improvement over the $\Omega(n\log n)$ barrier, but only for homogeneous circuits. While a noncommutative circuit can always be homogenized, this incurs a multiplicative $d^2$ blowup in size, where $d$ is the degree of the output polynomial, and so homogeneous circuit lower bounds do not immediately translate to lower bounds for general circuits. In an attempt to go beyond homogeneity, the first author recently proved \cite{shastri2026} lower bounds better than $\Omega(n\log n)$ for noncommutative circuits with low syntactic degree. This is a class of circuits that generalizes homogeneous circuits. To explain in more detail, we set up some notation.

Let $\F$ be a field, let $X = \{x_1,\dots,x_n\}$ and $Y = \{y_1,\dots,y_n\}$ be disjoint sets of noncommuting variables, and let $\Pal_{n,d}(X,Y) \in \F\langle X,Y\rangle$ denote the palindrome polynomial
\[
	\Pal_{n,d}(X,Y) = \sum_{(i_1,\dots,i_d)\in[n]^d} \left(\prod_{j=1}^{d} x_{i_j}\right)\left(\prod_{j=1}^{d} y_{i_{d+1-j}}\right).
\]
A version of this polynomial already appears in Nisan's work as a witness separating formulas from circuits. On the one hand, $\Pal_{n,d}(X,Y)$ admits a simple homogeneous circuit of size $O(nd)$, obtained from the recursion $\Pal_{n,d} = \sum_{i=1}^{n} x_i \Pal_{n,d-1}\, y_i$. On the other hand, Nisan's formula lower bound implies that any formula computing $\Pal_{n,d}(X,Y)$ requires size $n^{\Omega(d)}$.

As alluded to before, in a previous work, the first author proved a lower bound for circuits computing $\Pal_{n,n}$ under certain restrictions on the syntactic degree of the circuit. The syntactic degree of a gate is defined recursively as follows: a leaf labeled by a variable has syntactic degree 1, and a leaf labeled by a field constant has syntactic degree 0. For a sum gate, the syntactic degree is the maximum of the syntactic degrees of its children, while for a product gate, it is the sum of the syntactic degrees of its children. The syntactic degree of the circuit is the maximum syntactic degree of any gate. It is well known that if a circuit is homogeneous and computes a polynomial of degree $d$, then, without loss of generality, its syntactic degree may be assumed to be exactly $d$. In \cite{shastri2026}, the first author proved the following:

\begin{theorem}\label{thm:prev}
	Every fan-in 2 noncommutative arithmetic circuit computing $\Pal_{n,n}(X,Y)$ with syntactic degree $O(n)$ has size $\Omega(n^{1+c})$ for some constant $c > 0$.
\end{theorem}

The proof of Theorem~\ref{thm:prev} proceeds by constructing a ``small'' spanning set for the coefficient space of $\Pal_{n,n}$, starting from a small circuit. The current paper uses similar techniques and improves upon \cite{shastri2026} to prove a tight lower bound for all circuits computing $\Pal_{n,d}(X,Y)$, without any restriction on the syntactic degree.

\begin{theorem}\label{thm:main-intro}
	Every fan-in 2 noncommutative arithmetic circuit computing $\Pal_{n,d}(X,Y)$ has size $\Omega(nd)$.
\end{theorem}

This then shows that the \(O(nd)\) circuit is tight up to a constant.

In the analysis of \(\Pal_{n,d}(X,Y)\) we heavily use the assumption,
that it is defined using two sets of variables, where every \(X\)-variable
appears before every \(Y\)-variable.
In \Cref{thm:replacevariable} we show that this assumption can be removed,
and that our lower bound works for any polynomial with a full rank
coefficient matrix.

\subsection{AI Acknowledgement}\label{sec:ai}
We used Gemini AI to assist with removing the assumption of the circuit being low-syntactic degree.
In particular, it found that when we considered the degree \(d\) projection of our measure, we can split polynomials into constant and non-constant parts, and then discard every term which is a product of more than \(d\) non-constant parts.
the The tool materially affected \Cref{sec:final}. We further used AI for proofreading. The authors verified the correctness and originality of all content, including references.

\subsection{Comparison with concurrent work}

In concurrent work, Raz \cite{raz2026} proved an $\Omega(d\sqrt{n})$ lower bound for explicit $n$-variate degree $d$ polynomials, and also showed that such lower bounds imply lower bounds for arithmetic circuits over free noncommutative rings.

Raz's proof uses the complexity measure defined
in \cite{nisan1991} given by the rank of a matrix defined by the
coefficients of a noncommutative polynomial. He then relates the complexity
of the polynomial computed by a vertex to that of its children
and then considers a worst-case path from the root to a leaf, to get a lower
bound.

Our proof follows a similar strategy, where we use a variant of the same measure,
but instead consider two sets of variables \(X\) and \(Y\) that simplifies the analysis. Furthermore, we use the fact that the measure only depends on at most \(d\) product gates along a path from the root to a leaf.
This then results in a stronger lower bound, with a shorter analysis.

Raz also proved an $\Omega(d\sqrt{n})$ functional lower bound for polynomial functions over the ring $R = \mathbb{F}\langle Z \rangle$ \cite[Corollary 3]{raz2026} via a transfer theorem (\cite[Theorem 2]{raz2026}). We observe that his transfer theorem can be used to obtain an $\Omega(nd)$ functional lower bound over $R$ (Section \ref{sec:functional}).
\subsection{Future work}
As mentioned in the introduction, a superlinear lower bound for constant degree polynomials would imply exponential lower bounds.
For constant degree our result only implies a linear bound,
so the obvious goal is to get a stronger lower bound.

Since our complexity measure already achieves the optimal bound for the
palindrome polynomial,
a stronger lower bound would require developing new techniques.
A potential candidate for a polynomial is
the lifted identity polynomial which is conjectured to be hard to compute \cite{hrubes2011}
\[
	\LID_{n,d}(X,Y) =
	\sum_{(i_1,\dots,i_d)\in[n]^d} \left(\prod_{j=1}^{d} x_{i_j}\right)\left(\prod_{j=1}^{d} y_{i_{j}}\right).
\]
Note that this polynomial does not have the same recursive structure as the palindrome
polynomial.
A strategy to stronger lower bound would then be to find a complexity measure that can distinguish
between the palindrome polynomial and the lifted identity polynomial.

\section{Proof idea and notation}

We now set up the notation used in the proof and describe the proof idea. We continue to use the coefficient space viewpoint from \cite{shastri2026}, which we describe below. These ideas stem originally from Nisan's work \cite{nisan1991}. A polynomial $f \in \F\langle X,Y\rangle$ is said to be $X,Y$-separated if in each nonzero monomial of $f$, every $X$-variable appears before every $Y$-variable. The palindrome polynomial $\Pal_{n,d}(X,Y)$ is $X,Y$-separated. Accordingly, if $f$ is $X,Y$-separated, we may write $f = \sum_{m\in Y^*} c_m\, m$, where $c_m \in \F\langle X\rangle$. We write $\coeff_m\left(f\right) \coloneqq c_m$, and we define $\CoeffX(f) \coloneqq \{\coeff_m\left(f\right) \mid m \in Y^*,\ c_m \not\equiv 0\} \subseteq \F\langle X\rangle$, where \(Y^{*}\) denotes the set of non-constant monomials in the \(Y\)-variables.
For $f = \Pal_{n,d}(X,Y)$, the set $\CoeffX(f)$ is exactly the set of all degree $d$ monomials in the variables $X$. Since distinct noncommutative monomials are linearly independent, it follows that $\dim_{\F} \spn\left\{\CoeffX(\Pal_{n,d}(X,Y))\right\} = n^d$.

The rest of the proof is an upper bound on this same dimension in terms of the size of a circuit computing $\Pal_{n,d}(X,Y)$. Let $\Psi$ be a fan-in 2 circuit of size $s$ computing $\Pal_{n,d}(X,Y)$.
The spanning-set construction from \cite{shastri2026} associates to $\Psi$ a family of sets $\B_0 \subseteq \B_1 \subseteq \cdots \subseteq \B_{d'} \subseteq \F\langle X\rangle$ such that $\CoeffX(\Pal_{n,d}(X,Y)) \subseteq \spn(\B_{d'})$.
In the present proof, we continue to work with a similar construction,
except that the sets \(\B_{v}\) are indexed by the vertices of the circuit
and $\CoeffX(\Pal_{n,d}(X,Y)) \subseteq \spn(\B_{r})$
where \(r\) is the root of the circuit. This change is mostly for
simplifying the proof, as we no longer assume a
bound on the syntactic degree of the circuit.

In the previous paper, we bounded the size of the spanning sets themselves.
That is no longer the right quantity here.
For general circuits, the set $\B_{r}$ can be too large.
The new step is to apply the linear projection $\pi_d : \F\langle X\rangle \to \F\langle X\rangle$ onto the homogeneous degree $d$ component and count only those elements of $\B_{r}$ whose image under $\pi_d$ is nonzero.

The key observation is that every
element of \(\B_{r}\) corresponds to a sequence of product gates strictly ordered by
depth as described in \Cref{lem:structure}.
After expanding, the minimum degree of each term is lower bounded by the length $\ell$ of the corresponding sequence,
so it follows that if $\ell > d$, then its degree $d$ projection is zero. So, after projection, only sequences of length at most $d$ can survive.

At this point the problem becomes a counting problem.
Assuming that the circuit contains \(s\) product gates,
the number of valid sequences of product gates of length \(\ell \leq d\)
is at most \(\binom{s}{\ell}\), since the order is fixed by the depth.
Hence \(\pi_{d}(\B_{r})\) is spanned by a set of size at most
$\sum_{\ell=0}^{d} \binom{s}{\ell}$. Since $\CoeffX(\Pal_{n,d}(X,Y))$ lies in the span of these projections, we obtain $n^d \le \sum_{\ell=0}^{d} \binom{s}{\ell}$.
An estimate now shows that this inequality is impossible when $s = o(nd)$. This yields the required lower bound of $s = \Omega(nd)$.

We now set up some more notation used in the proof. For any nonzero polynomial $f \in \F\langle X\rangle$, define $\mindeg(f)$ to be the minimum degree of a monomial appearing in $f$. The following lemma is immediate:

\begin{lemma}\label{lem:mindeg}
	For all nonzero $f, g \in \F\langle X\rangle$, $\mindeg(fg) = \mindeg(f) + \mindeg(g)$.
\end{lemma}

We allow our circuits to have edges labeled by scalars from the field. For concreteness, we define our circuit model below.

\begin{definition}[Circuit model]
	A fan-in 2 noncommutative arithmetic circuit is a directed acyclic graph whose leaves are labeled by field constants and variables and whose internal nodes (called gates) are labeled by $+$ or $\times$. Every internal node has in-degree 2. The $\times$ gates all have designated left and right children. Each gate computes a noncommutative polynomial in the natural way. We allow edges labeled by field constants, these simply scale the polynomial carried by the edge. The size of the circuit is the total number of vertices in its graph. The depth of a gate $g$ is the length of the longest path from a leaf to $g$.
\end{definition}

\section{\texorpdfstring{Construction of spanning set from \cite{shastri2026}}{Construction of spanning set from [Sha26]}}

In this section, we recall and slightly modify the construction of spanning sets from \cite{shastri2026}. Let $\Psi$ be a fan-in 2 noncommutative arithmetic circuit of size $s$ computing $\Pal_{n,d}(X,Y)$.

For each gate $g$ of $\Psi$, let $\h{g}$ denote the polynomial computed at $g$.




For each gate $g$, let
\begin{itemize}
	\item $\h{g}_X$ be the sum of all monomials of $\h{g}$ involving only variables from $X$
	\item $\h{g}_Y$ be the sum of all monomials of $\h{g}$ involving only variables from $Y$
	\item $\h{g}_{XY}$ be the sum of all \(X,Y\)-separated monomials of
	      $\h{g}$
\end{itemize}


\begin{definition}
	For every vertex \(v\) of \(\Psi\), define a set \(\mathcal{B}_{v} \subseteq \F\langle X \rangle\)
	inductively as follows.
	\begin{itemize}
		\item \textbf{Leaves:}
		      $\mathcal{B}_{\ell} = \left\{1\right\}$.
		\item \textbf{Sum gates:} $\mathcal{B}_{v} = \bigcup_{u \in \ch(v)} \mathcal{B}_{u}$.
		\item \textbf{Product gates:} $p = a \times b$.
		      $\mathcal{B}_{p} = \mathcal{B}_{a} \cup \left\{\h{a}_X \cdot h \mid h \in \mathcal{B}_{b}\right\}$.
	\end{itemize}
\end{definition}

\begin{claim}\label{claim:span}
	For every gate $g$, $\CoeffX\left(\h{g}_{XY}\right) \subseteq \spn(\B_g)$.
\end{claim}

\begin{proof}
	We do induction on the structure of the circuit.

	\textbf{Leaves.}
	Since \(\hat{\ell}\) has degree at most 1,
	the coefficient of any non-constant \(Y\)-monomial is a scalar,
	so \(\CoeffX(\hat{\ell}) \subseteq \spn \{1\}\).

	\textbf{Sum gates.}
	Since $\h{v} = \sum_{u \in \ch(v)} \h{u}$, linearity gives for every
	\(Y\)-monomial \(m\)
	\[
		\coef_{m}(\h{v}) \; = \sum_{u \in \ch(v)} \coef_{m}(\h{u}).
	\]
	By the inductive hypothesis $\coef_{m}(\h{u}) \in \spa(\mathcal{B}_u)$ for each $u$,
	so
	\[
		\coef_{m}(\h{v})
		\;\in\; \spa \Bigl(\bigcup_{u \in \ch(v)} \mathcal{B}_u\Bigr)
		= \spa(\mathcal{B}_v).
	\]
	\textbf{Product gates.}
	Let $p = a \times b$, so $\h{p} = \h{a} \cdot \h{b}$.
	Then we can write
	\begin{align*}
		\h{p}_{XY} =
		\h{a}_{X} \h{b}_{XY}
		+
		(\h{a}_{XY}-\h{a}_{X}) \h{b}_{Y}.
	\end{align*}
	By assumption \(\CoeffX(\h{b}_{XY}) \subseteq \spn \, \mathcal{B}_{b}\),
	so
	\(\CoeffX(\h{a}_{X}\h{b}_{XY})
	\subseteq
	\spn \left( \h{a}_{X} \cdot \mathcal{B}_{b}\right)\).
	Similarly
	\(\CoeffX(\h{a}_{XY}-\h{a}_{X}) = \CoeffX(\h{a}_{XY}) \subseteq \spn \, \mathcal{B}_{a}\),
	so
	\(\CoeffX((\h{a}_{XY}-\h{a}_{X})\h{b}_{Y}) \subseteq \spn \, \mathcal{B}_{a}\),
	since every coefficient of a non-empty \(Y\)-monomial
	is a linear combination of the coefficients of non-empty \(Y\)-monomials
	in \(\h{a}_{XY}\).
	We then have
	\begin{align*}
		\CoeffX(\h{p}_{XY}) \subseteq
		\spn \left\{ \mathcal{B}_{a} \cup
		\left\{ \h{a}_{X} \cdot \mathcal{B}_{b}\right\}\right\},
	\end{align*}
	completing the induction.
\end{proof}

We now show that the inductive structure of the sets \(\mathcal{B}_{v}\)
leads to a simple description of their elements.

\begin{lemma}\label{lem:structure}
	For every vertex $v$ we have
	\[
		\mathcal{B}_v =
		\bigl\{\, w_\pi \mid
		\pi \text{ is a path from } v \text{ to a leaf}\,\bigr\},
	\]
	where for a path $\pi = (v_0, v_1, \ldots, v_m)$
	with $v_0 = v$ and $v_m$ a leaf,
	\[
		w_\pi = \prod_{i=0}^{m} c_{v_i},
	\]
	where
	\[
		c_{v_i} =
		\begin{cases}
			(\h{a}_{i})_X & \quad \text{if} \quad v_i = a_{i} \times b_{i} \text{ is a product gate and } v_{i+1} = b_{i}, \\
			1             & \quad \text{otherwise.}
		\end{cases}
	\]
\end{lemma}

\begin{proof}
	We proceed by induction on the structure of the circuit.

	\textbf{Leaves:}
	The only path from $\ell$ to itself has $m=0$, so $w_\pi = 1$.

	\textbf{Sum gates:}
	Every path from $v$ corresponds to a path from one of its children $u \in \ch(v)$
	with the same product gates, so
	by the inductive hypothesis, $\{w_\pi \mid \pi \text{ from } v\} = \bigcup_{u \in \ch(v)} \mathcal{B}_u = \mathcal{B}_v$.

	\textbf{Product gates:}
	Let $p = a \times b$; every path from $p$ begins with $p \to a$ or $p \to b$.
	Let $\pi|_a$ denote the restriction of \(\pi\)
	so that it starts at \(a\), and similar for \(\pi|_{b}\).

	\emph{Left branch.}
	The path takes the left branch at $p$, so $p$ contributes no factor and
	$w_\pi = w_{\pi|_a}$. By the inductive hypothesis $w_{\pi|_a} \in \mathcal{B}_a$.

	\emph{Right branch.}
	The path takes the right branch at $p$, so $p$ is the first right-branch gate and
	$w_\pi = \h{a}_X \cdot w_{\pi|_b}$.
	By the inductive hypothesis $w_{\pi|_b} \in \mathcal{B}_b$,
	so $w_\pi \in \{\h{a}_X \cdot h \mid h \in \mathcal{B}_b\}$.

	The two cases together yield $\mathcal{B}_p = \mathcal{B}_a \cup \{\h{a}_X \cdot h \mid h \in \mathcal{B}_b\}$.
\end{proof}
\begin{remark}
	From \Cref{lem:structure},
	we can already directly get an exponential lower bound for formulas by
	observing that in a tree,
	there is a unique path from the root to each leaf,
	so we get \(|\mathcal{B}_{r}| \leq L\),
	where \(L\) is the number of leaves.
	So if \(\h{r} = \Pal_{n,d}\),
	then from Claim~\ref{claim:span} we get
	\begin{align*}
		n^{d} = \dim_{\F} \spn \{\CoeffX(\Pal_{n,d})\}
		\leq | \mathcal{B}_{r} | \leq L.
	\end{align*}
	This recovers Nisan's lower bounds for formulas~\cite{nisan1991}.

	For general circuits, we cannot get optimal lower bounds by just bounding
	the number of paths from the root. This number is only bounded above, in the worst case, by \(2^{S}\), where \(S\) is the number of product gates, and this only implies an \(\Omega(d \log n)\) lower bound.
	To get the full \(\Omega(nd)\) lower bound, we need to consider the projections to degree \(d\), as we do in the next section.
\end{remark}
\section{Projection and Counting Projected Spanning Sets}\label{sec:final}

In this section, we bound the size of $\pi_d(\B_{r})$, the size of the (projected) spanning set associated with the root. That is, we would like to project the set $\B_{r}$ to homogeneous degree $d$ components via the linear map $\pi_d$ and count only those elements in $\B_{r}$ whose degree $d$ projections survive.
\begin{lemma}\label{lem:bound}
	Let $\Psi$ be a circuit with root \(r\) and at most \(S\) product gates.
	Then \(\pi_{d}(\mathcal{B}_{r})\)
	is spanned by a set of size at most
	\begin{align*}
		\sum_{k=0}^{d} \binom{S}{k}.
	\end{align*}
	%
\end{lemma}

\begin{proof}
	By Lemma~\ref{lem:structure},
	we know that every element of \(\mathcal{B}_{r}\)
	can be written in the form
	$w_{\pi} = \h{(a_1)}_X\cdots\h{(a_k)}_X$,
	where the \(a_{i}\) are left vertices of product gates ordered by depth along a path from
	\(r\) to a leaf \(\ell\).
	Let \(\h{(a_i)}_{X} = \h{(a_i)}_{>0} + \h{(a_i)}_{0}\)
	be the decomposition into its non-constant and constant terms.
	Expanding out \(w_{\pi}\) in this decomposition we get
	\begin{align*}
		w_{\pi} = \sum_{T \subseteq [k]} c_{T} \prod_{i \in T} \h{(a_i)}_{>0},
	\end{align*}
	where \(c_{T}\) is the product of all the constant terms of indices not
	in \(T\). Since \(\h{(a_i)}_{>0}\) has minimal degree at least \(1\),
	every term with \(|T| > d\) has minimal degree more than \(d\)
	so it vanishes when taking the projection \(\pi_{d}\).
	It follows that every \(\pi_{d}(w_{\pi})\) is in the span of the
	projection of the set
	\begin{align*}
		\left\{
		\prod_{i \in T} \h{(a_i)}_{>0} \; | \; T \subseteq [S], |T| \leq d \right\},
	\end{align*}
	which has cardinality \(\sum_{k=0}^{d} \binom{S}{k}\).
\end{proof}

\begin{theorem}\label{thm:main}
	Let $n \geq 2$ and $d \ge 1$, and let $\Psi$ be a fan-in 2 noncommutative arithmetic circuit computing $\Pal_{n,d}(X,Y)$.
	Then $\Psi$ has $S = \Omega(nd)$ many product gates.
\end{theorem}

\begin{proof}
	Suppose $d\leq S/2$. In this case, by comparing the bound in Lemma~\ref{lem:bound} with the dimension of the coefficient space
	$\dim_\mathbb{F}\spn\{\CoeffX(\Pal_{n,d})\}= n^d$, we see that

	\begin{align*}
		n^d \leq \sum_{k=0}^d \binom{S}{k} \leq (d+1)\binom{S}{d}
		\leq (d+1)\left(\dfrac{eS}{d}\right)^d
	\end{align*}

	Therefore, $S\geq \frac{nd}{e(d+1)^{1/d}}\geq (nd)/6$, and we're done.

	Now suppose $S/2<d$, then $S \leq 2d$ and the bound in Lemma~\ref{lem:bound} gives $$n^d \leq \sum_{k=0}^{d}\binom{S}{k} \leq (d+1)\binom{S}{\lfloor S/2\rfloor} \leq (d+1)2^{S},$$

	and therefore $S \geq d \log n - \log (d+1)$. Also, if $S/2<d$ we additionally have $n^d \leq (d+1) \binom{2d}{d}\leq (d+1)2^{2d}$. This implies $n \leq 5$ and so $S \geq d\log n - O(\log d)= \Omega(nd)$ in this case as well (we utilize $n\geq 2$ here).
\end{proof}

\section{Reducing to one set of variables}

In this section, we observe that the constraint of $X,Y$-separatedness we imposed in the previous sections can be removed. In particular, we note that a strengthening of Raz's lower bound (Raz~\cite{raz2026}, Theorem 1) follows from our techniques. Let $X = \{x_1, \ldots, x_n\}$, let $f\in\mathbb{F}\langle X\rangle$ have degree $2d$, an even number. Let $M_f^{d,d}$ be the $n^{d}\times n^{d}$ matrix whose rows and columns are indexed by monomials in $X^{d}$ and whose $(m,m')$-th entry is the coefficient of $mm'$ in $f$. Notice that if $f$  were $X,Y$-separated and bihomogeneous of bidegree $(d,d)$ then $\rank(M_f^{d,d}) = \dim(\spn\{\CoeffX(f)\})$. We will utilize this observation in the subsequent proof.

\begin{theorem}[cf. Theorem 1, \cite{raz2026}]\label{thm:replacevariable}
	Let $n\geq 2$ and $d\geq 1$.
	Let $X=\{x_1, \ldots, x_n\}$. If $f\in\mathbb{F}\langle X\rangle$ is such that $M_f^{d, d}$ has full rank, then the size of any circuit computing $f$ is $\Omega(nd)$.
\end{theorem}

\begin{proof}
	Let $C$ be a circuit computing $f$. Perform the substitution $x_i\leftarrow x_i + y_i$ for all $i\in[n]$ and obtain a new circuit $C'$.
	Notice that the bidegree $(d, d)$ $X,Y$-separated part of the polynomial computed by $C'$ is exactly the $X,Y$-separated version of $f$.
	That is, it is $f$ but with the $x_i$'s in every length $d$ suffix replaced by the corresponding $y_i$'s. Call this new polynomial $\tilde{f}$. Construct a basis for the coefficient space of  $\tilde{f}_{XY}$ as in Claim~\ref{claim:span} and apply the projection $\pi_{d}$ to both the space of coefficients of $\tilde{f}_{XY}$ and also to the basis we constructed. Note that
	$\dim(\spn\{\pi_{d}(\CoeffX(\tilde{f}_{XY}))\}) = n^{d}$.

	By the analysis in Lemma~\ref{lem:bound} and Theorem~\ref{thm:main}, the size of $C'$ is $\Omega(nd)$ and therefore so is the size of $C$. \end{proof}

\section{On functional lower bounds}\label{sec:functional}

As we mentioned before, Raz \cite{raz2026} showed a transfer result and proved a lower bound of $\Omega(d\sqrt{n})$ for functional noncommutative circuits over the ring $R = \mathbb{F}\langle z_1, \ldots z_n\rangle$. In particular, he showed that if $f(x_1, \ldots, x_n)$ has a noncommutative circuit of size $s$ as a function $f:R^n \rightarrow R$ (the circuit is allowed to have coefficients in $R$), then $f$ also has a symbolic noncommutative circuit of size $s$ computing it over $\mathbb{F}$ \cite[Theorem 2]{raz2026}. Using this, he obtained the $\Omega(d\sqrt{n})$ functional lower bound mentioned above. We observe that \cite[Theorem 2]{raz2026} combined with Theorem \ref{thm:replacevariable} allows us to obtain an $\Omega(nd)$ lower bound.

\begin{corollary}[{cf. \cite[Corollary 3]{raz2026}}]
	Let $f\in \mathbb{F}\langle X \rangle$ be degree $d$ such that the matrix $M^f_{d/2,d/2}$ has full rank. Let $C$ be an arithmetic circuit with inputs $\{x_1, \ldots, x_n\}$ over the ring R which computes an $f(x_1,...,x_n)$ as a function from $R^n \to R$. Then, there are at least $\Omega(nd)$ product gates in C. \hfill $\square$
\end{corollary}

\section{Acknowledgements}

The first author would like to thank C. Ramya for fruitful and interesting discussions about this problem.

\bibliographystyle{alpha}
\bibliography{ref}

@article{baur1983,
  author    = {Walter Baur and Volker Strassen},
  title     = {{The complexity of partial derivatives}},
  journal   = {Theoretical Computer Science},
  volume    = {22},
  number    = {3},
  pages     = {317--330},
  year      = {1983}
}

@inproceedings{carmosino2018,
  author    = {Marco L. Carmosino and Russell Impagliazzo and Shachar Lovett and Ivan Mihajlin},
  title     = {{Hardness Amplification for Non-Commutative Arithmetic Circuits}},
  booktitle = {33rd Computational Complexity Conference (CCC 2018)},
  editor    = {Thomas Eiter and David Sands},
  series    = {Leibniz International Proceedings in Informatics (LIPIcs)},
  volume    = {102},
  pages     = {12:1--12:16},
  address   = {Dagstuhl, Germany},
  publisher = {Schloss Dagstuhl -- Leibniz-Zentrum f{\"u}r Informatik},
  year      = {2018}
}

@inproceedings{chatterjee2023,
  author    = {Prerona Chatterjee and Pavel Hrube{\v{s}}},
  title     = {{New Lower Bounds Against Homogeneous Non-Commutative Circuits}},
  booktitle = {38th Computational Complexity Conference (CCC 2023)},
  editor    = {Amnon Ta-Shma},
  series    = {Leibniz International Proceedings in Informatics (LIPIcs)},
  volume    = {264},
  pages     = {13:1--13:10},
  address   = {Dagstuhl, Germany},
  publisher = {Schloss Dagstuhl -- Leibniz-Zentrum f{\"u}r Informatik},
  year      = {2023}
}

@article{hrubes2011,
  author    = {Pavel Hrube{\v{s}} and Avi Wigderson and Amir Yehudayoff},
  title     = {{Non-commutative circuits and the sum-of-squares problem}},
  journal   = {Journal of the American Mathematical Society},
  volume    = {24},
  number    = {3},
  pages     = {871--898},
  year      = {2011}
}

@inproceedings{nisan1991,
  author    = {Noam Nisan},
  title     = {{Lower Bounds for Non-Commutative Computation}},
  booktitle = {Proceedings of the Twenty-Third Annual ACM Symposium on Theory of Computing},
  editor    = {Shafi Goldwasser},
  pages     = {410--418},
  publisher = {ACM},
  year      = {1991}
}

@misc{raz2026,
  author    = {Ran Raz},
  title     = {{Polynomial Lower Bounds for Arithmetic Circuits over Non-Commutative Rings}},
  howpublished = {arXiv:2604.22006, ECCC TR26-061},
  url       = {https://arxiv.org/abs/2604.22006},
  note      = {To appear in FOCS 2026},
  year      = {2026}
}

@inproceedings{shastri2026,
  author    = {Pratik Shastri},
  title     = {{Lower Bounds for Noncommutative Circuits with Low Syntactic Degree}},
  booktitle = {17th Innovations in Theoretical Computer Science Conference (ITCS 2026)},
  editor    = {Shubhangi Saraf},
  series    = {Leibniz International Proceedings in Informatics (LIPIcs)},
  volume    = {362},
  pages     = {115:1--115:9},
  address   = {Dagstuhl, Germany},
  publisher = {Schloss Dagstuhl -- Leibniz-Zentrum f{\"u}r Informatik},
  year      = {2026}
}

@article{strassen1973,
  author    = {Volker Strassen},
  title     = {{Die Berechnungskomplexit{\"a}t von elementarsymmetrischen Funktionen und von Interpolationskoeffizienten}},
  journal   = {Numerische Mathematik},
  volume    = {20},
  number    = {3},
  pages     = {238--251},
  month     = {Jun},
  year      = {1973}
}

\end{document}